\documentclass[pdflatex,sn-mathphys-num]{sn-jnl}

\usepackage{mathtools}
\usepackage{graphicx} %
\usepackage{tikz}

\usepackage{amssymb}
\usepackage{physics}
\usepackage[shortlabels]{enumitem}
\usetikzlibrary{decorations.pathmorphing}

\usepackage{amsmath,amssymb,amsthm}

\newtheorem{theorem}{Theorem}[section]
\newtheorem{lemma}[theorem]{Lemma}
\newtheorem{proposition}[theorem]{Proposition}
\newtheorem{corollary}[theorem]{Corollary}

\theoremstyle{definition}

\theoremstyle{remark}
\newtheorem{remark}[theorem]{Remark}

\numberwithin{equation}{section}

\newcommand{\bfeta}{\boldsymbol{\eta}}

\begin{document}

\title{Mixed order phase transition in a locally constrained exclusion process}

\author[1]{\fnm{Gunter M.} \sur{Sch\"utz}}

\author*[2,3]{\fnm{Ali} \sur{Zahra}
}\email{ali.zahra@univ-lorraine.fr}

\affil*[1]{\orgdiv{Institute for Advanced Simulation 2}, 
  \orgname{Forschungszentrum J\"ulich}, 
  \orgaddress{\city{J\"ulich}, \postcode{52424}, \country{Germany}}}

\affil[2]{\orgname{Laboratoire de Physique et Chimie Th\'eoriques, Universit\'e de Lorraine},
  \orgaddress{\city{Nancy}, \country{France}}}

\affil[3]{\orgname{Instituto Superior T\'ecnico, Universidade de Lisboa},
  \orgaddress{\city{Lisbon}, \country{Portugal}}}
\date{\today}

\abstract{
We investigate a novel variant of the exclusion process in which particles perform asymmetric nearest-neighbor jumps across a bond \((k, k+1)\) only if the preceding site \((k-1)\) is unoccupied. This next-nearest-neighbor constraint significantly enriches the system’s dynamics, giving rise to long-range correlations and a mixed-order transition controlled by the asymmetry parameter.
We focus on the critical case of half filling, where the system splits into  two ergodic components, each associated with an invariant reversible measure. The  combinatorial structure of this equilibrium distribution is intimately connected to the \(q\)-Catalan numbers, enabling us to derive rigorously the asymptotic behavior of key thermodynamic quantities in the strongly asymmetric regime and to conjecture their behavior in the weakly asymmetric limit.Even though the system is one-dimensional and has short-range interactions, an equilibrium phase transition occurs between a clustered phase -- characterized by slow dynamics, long-range correlations with thermodynamic additivity, and spontaneous breaking of translational symmetry -- and a  fluid phase where the correlations are short-range and which is thermodynamically additive. This equilibrium phase transition features characteristics of a first-order transition, such as a discontinuous order parameter as well as characteristics of a second-order transition, namely a divergent susceptibility at the transition point. We also briefly discuss density higher than one half where ergodicity is broken.
}

\maketitle

\section{Introduction}
Kinetically constrained lattice gases (KCLG) are interacting particle systems
which show some of the key features of ergodicity breaking transitions and glassy dynamics \cite{helbing2000simulating}
In such systems, particles can hop on a lattice only if certain local conditions are fulfilled. Such constraints can generate large-scale effects, including spatial heterogeneities, metastability, aging, ergodicity breaking, and other glass-like behavior \cite{ritort2003glassy,garrahan2011kinetically}.

An early influential example is the  Kob-Andersen (KA) model on a three-dimensional cubic lattice \cite{kob1993kinetic} with exclusion interaction, i.e., allowing for at most one particle per site. 
A particle may hop to a neighboring vacant site only if {\it both} the departure and the arrival sites possess at least \(m\) empty nearest neighbors  before and  after the move (with the canonical choice \(m=3\)). 
This local “cage” constraint reproduces the cooperative crowding observed in dense liquids: As the global density increases, more particles become blocked by their neighbors and the structural relaxation time grows sharply. 
Thus the model captures the paradigmatic glassy scenario in which the configuration remains disordered, like a liquid, yet the dynamics are arrested, like a solid.

More generally, multi–site kinetic constraints arise in many physical contexts in one or more dimensions. Well-known examples include cooperative/extended-particle processes such as \(k\)-mer random sequential adsorption \cite{evans1993random}; porous media models, where various one-dimensional exclusion processes with rates depending on next–nearest occupancies have been shown to converge, in hydrodynamic scaling, to (fractional) porous–medium equations \cite{GoncalvesLandimToninelli2009,Goncalves2011,CardosoGoncalves2024}, driven transport with extended particles, most notably ribosome traffic, where a ribosome of footprint \(\ell>1\) advances only if a block of \(\ell\) empty sites is available ahead \cite{macdonald1968biopolymers,Schu97a,Laka03,shaw2003totally,chou2011rpp}, systems of cold-atom Rydberg gases, where kinetic constraints appear naturally in
the effective evolution equation of the Rydberg systems leading in the limit of strong dissipation to a glassy dynamics and slow relaxations.
\cite{lesanovsky2013kinetic_constraints,valado2016constraints}, granular “Tetris-like’’ compaction, where rearrangements need compatible multi-site voids \cite{caglioti1997tetris}; and in cellular automata for traffic flow, where a vehicle advances only if a multi-cell gap is free \cite{nagel1992traffic}.

A particular kinetically constrained lattice model that has attracted attention in recent years is the facilitated exclusion process (FEP)\cite{Ross00,basu2009active,gabel2010facilitated,Ayye23,Lei23}. This process was originally introduced to study absorbing phase transitions with a conserved field. Once the transient states have died out, the absorbing domain remains active and the facilitated exclusion process becomes equivalent to an exclusion process with extended particles.
This is a fluid-like system with short-range stationary correlations for which the hydrodynamic continuity equation was originally derived in \cite{Scho04}
and recently proved rigorously \cite{Blon20,Blon21}.

Models with multi–site kinetic constraints are generally divided into two broad classes: 
{\it cooperative} and {\it non-cooperative}, \cite{toninelli2005cooperative,shapira2024noncooperative}.
In a cooperative model, a particle can move only when {\it all} of a specified set of neighbouring sites satisfy the constraint; if any one of them fails, the move is strictly forbidden. 
By contrast, a non-cooperative model merely  modulates the local transition rate -- slowing it down or speeding it up -- without ever prohibiting the move outright.

A simple realization of these two scenarios in one dimension is the model of Ref. \cite{helbing1999global}.  
A particle hops forward with rate \(r\) when it is {\it preceded} by another particle, (i.e. $110 \;\xrightarrow{\,r\,}\; 101$)
whereas if it is preceded by a hole, it hops forward with a rate $q$ ( i.e. $010 \; \xrightarrow{\,q\,}\; 001$). The latter update rule can be reversed with rate 1 (i.e. $001 \; \xrightarrow{\,1\,} \; 010$)

For \(r,q>0\) every move remains possible, so the process is non-cooperative.  
If either rate is zero, certain moves become impossible unless neighbouring particles collaborate, and the dynamics enters a genuinely cooperative (or “facilitated’’) regime.  As pointed out above, 
such cooperative rules typically produce qualitatively different behaviour -- slow relaxation, dynamical heterogeneity, and glass-like arrest.  
In the present work, we study the cooperative limit in which forward motion is possible  only when a particle is preceded by a hole, corresponding to the  
choice \(r=0\) 
while keeping the asymmetry parameter \(q>0\)
arbitrary.  The fully constrained limit \(r=q=0\) maps, via particle–hole exchange, onto the totally asymmetric facilitated exclusion process.

Eliminating the \(110\!\to\!101\) transition maximally amplifies the cooperative character of the dynamics: a particle can advance only if its rear neighbour vacates first. As shown below, this produces pronounced caging and slow relaxation reminiscent of glass formers.  
The remaining control parameter \(q\) then tunes the system smoothly from a condensed phase where additive thermodynamics break down with extensive quantities such as the free energy becoming superextensive signaling effective long-range correlations ($q>1$), to a fluid-like phase where the system behaves as an (ordinary) lattice gas with short-range correlation ($q<1$).  At the transition point, the order parameter is
discontinuous as in a first order phase transition  while the susceptibility diverges at the transition point which is a characteristic property of second-order phase transition. This is a mixed order phase transition reminiscent of the one discussed in \cite{bar2014mixed}.
Our results are derived at half filling, where the model is reversible  and analytically tractable for any value of the asymmetry parameter $q$.

In Sec.\ref{sec1}, we define the model and review its main properties including the invariant measure at half filling. This section has already partially appeared in \cite{grosskinsky2025long}.
Sec. \ref{sec:qs-catalan}
develops the combinatorial backbone by defining the $(q,s)$–Catalan numbers that weight Dyck paths by area and number of returns, generalising thus the area $q$–Catalan numbers. We derive a $q$–functional equation and closed-form generating function, discuss special cases, and extract asymptotics (including the critical curve in $(q,s)$). Sec. \ref{sec:mixed} uses results obtained in the previous section to compute the asymptotic behaviour of key thermodynamic quantities such as the free energy, the internal energy and the susceptibility, allowing us to
characterize a mixed–order transition with discontinuous order parameter and divergant susceptibility at the transition point which is defined by the symmetric hopping $q=1$. 
Sec. \ref{sec:weakly} studies the weakly asymmetric regime $q_N=\exp(b/N^\alpha)$ and proposes the corresponding asymptotic behaviour of the different thermodynamic quantities.
Sec. \ref{sec:more} extends the analysis to densities above one half by reducing the dynamics to an active segment that again maps to Dyck paths, yielding explicit stationary weights and consequences for observables.

\section{Model and Main Properties}
\label{sec1}
We recall and elaborate on the key features introduced in \cite{grosskinsky2025long}. Consider a system of \(N\) particles on a one-dimensional lattice with \(L\) sites and periodic boundary conditions. The dynamics are governed by the local update rule:
\begin{equation}
    010 \underset{1}{\stackrel{q}{\rightleftarrows}} 001,
    \quad
    \text{with rate} \quad q \geq 0,
\end{equation}
where \(1\) denotes a particle and \(0\) denotes a hole. The special case \(q = 0\) corresponds to the facilitated Totally Asymmetric Exclusion   Process (F-TASEP) under hole-particle symmetry exchange. One inherited feature from F-TASEP is the existence of a critical density at half-filling.
For densities above half filling, the system splits into multiple disconnected ergodic components, as the number of available holes becomes critically low that some particles become permanently frozen, as we will explain in Sec. \ref{sec:more}. Consequently, the stationary current vanishes identically, even in finite systems and for any value of the asymmetry parameter. 

In contrast, for the number of particles strictly below half filling, the system is fully ergodic and contains no absorbing states—unlike F-TASEP, which exhibits absorbing states in this regime. In general, a non-zero stationary current is observed. An extensive analysis of the phenomenology in this regime is provided in \cite{grosskinsky2025long}.
At exactly half filling, we show that the system has two distinct ergodic components and that the stationary current vanishes identically.
Throughout the remainder of the paper (except in Sec.~\ref{sec:more}), we restrict ourselves to the half-filled case with an even system size, i.e., \(L = 2N\). We denote a configuration by \( \boldsymbol{\eta} = (\eta_1, \dots, \eta_L) \in A_N \), where \( \eta_k \in \{0, 1\} \) indicates the occupancy of site \(k\), and $A_N$ denotes the configuration space with $N$ particles. Define the spin variables as \( \sigma_k := 1 - 2\eta_k \), and introduce the height function:
$    h_k = \sum_{i=1}^{k} \sigma_i,$
with the convention \( h_0 = 0 \). Due to the half-filling condition, the height function satisfies \( h_L = 0 \).
Let \( h_{\text{min}} = \min_k h_k \) denote the minimum value of the height function.
The height function can be visualized as a path along the diagonal of a two-dimensional lattice: holes correspond to upward-left steps, and particles to downward-left steps, and the update rule can be stated in terms of the height path, Fig. ~\ref{fig:height}. Define a \textit{minimum site}  as a site \( k_0 \) where the height attains its minimum value, i.e., \( h_{k_0} = h_{\text{min}} \).
In the next lemma, we summarize the important properties of a half filled system.
\begin{figure}
    \centering
    \begin{tikzpicture}[yscale=1,xscale=1,scale = 0.6]
\draw [thick] (0,0) -- (1,1) -- (2,0) -- (3,1) ;
\draw [dashed] (0,0) -- (1,1) -- (2,2) -- (3,1) ;

\draw [->, thick] (3.7,1.1) -- (4.6,1.1);
\draw [<-, thick] (3.7,0.8) -- (4.6,0.8);

\node at (4.03,1.5) {q};
\node at (4.1,0.4) {1};

\draw [xshift = 5 cm,thick]    (0,0) -- (1,1) -- (2,2) -- (3,1) ;
\draw [xshift = 5 cm,dashed]    (0,0) -- (1,1) -- (2,0) -- (3,1) ;

\begin{scope}[yshift = -1cm]

\foreach \i in {0,...,2}
	{
    \draw (\i,0) -- (\i+1,0);
		\draw (\i,0) -- (\i,0.1) ;
	}
	
\foreach \i in {1}
	{
		\node at (\i+0.5,0.5){};
        \draw [very thick,fill] (\i+0.5,0.3) circle (5pt);
        }
\foreach \i in {0,2}
	{
		\node at (\i+0.5,0.5){};
        \draw [thick] (\i+0.5,0.3) circle (5pt);
        }

\end{scope}

\begin{scope}[yshift = -1cm, xshift = 5cm]
    
\foreach \i in {0,...,2}
	{
    \draw (\i,0) -- (\i+1,0);
		\draw (\i,0) -- (\i,0.1) ;
	}
	
\foreach \i in {2}
	{
		\node at (\i+0.5,0.5){};
        \draw [very thick,fill] (\i+0.5,0.3) circle (5pt);
        }
\foreach \i in {0,1}
	{
		\node at (\i+0.5,0.5){};
        \draw [thick] (\i+0.5,0.3) circle (5pt);
        }

\end{scope}

\end{tikzpicture}
    \caption{Local update in terms of the height function. The minimum is locally conserved by the dynamics}
    \label{fig:height}
\end{figure}
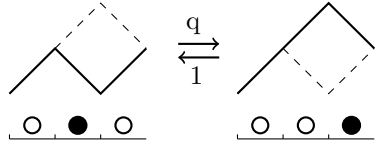

\begin{lemma}
\label{lemma1}
The half-filled system satisfies the following properties:

\begin{enumerate}[(a)]
    \item The dynamics preserve the minimum value of the height function.
    
    \item For a chosen configuration,  all sites where the height function attains its minimum share the same parity.(they are all even or odd sites), and this parity is preserved under the evolution.
    
    \item The configuration space $A_N$ decomposes into two connected ergodic components: one containing configurations with minimum sites at even positions, and the other with minimum sites at odd positions.
    By a slight abuse of terminology, we refer to these as the \textit{even} and \textit{odd} ergodic components, and denote them $\mathcal{A}_N^{(e)}$ and $\mathcal{A}_N^{(o)}$ respectivly.
    The two components are related by a single-site translation. Each contains exactly one antiferromagnetic configuration; \(010101\ldots01\) in the even component, and \(101010\ldots10\) in the odd component.
    
    \item Let \(J_t\) denote the integrated current, defined up to an additive constant as an integer-valued random variable that increases by one for each forward jump and decreases by one for each backward jump. Then \(J_t\) is only a function of the configuration with a minimum bound attained for an antiferromagnetic configuration, for which we set its value to zero. and an upper bound attained for any clustered configuration where all particles are consecutive, where its value is given by \(\frac{N(N - 1)}{2}\).
\end{enumerate}
\end{lemma}

\begin{proof}
Let \( k \) be a site where the minimum height occurs, i.e., \( h_{\text{min}} = h_k \). This implies that \( \eta(k) = 1 \) and \( \eta(k + 1) = 0 \). For the site \( k \) to cease being a minimum, the particle at \( k \) must jump to \( k+1 \). However, this can only occur if the site \( k-1 \) is vacant, which means \( h_{k-2} = h_k \). In other words, \( k-2 \) must also be a minimum site. Consequently, a minimum site \( k \) cannot be annihilated unless another minimum exists at \( k-2 \).

Similarly, it can be verified that a new minimum site \( l \) can only be created if \( l-2 \) is already a minimum site. This proves property (a).

To prove property (b), observe that between two minimum sites \( k \) and \( l \), the number of upward steps must equal the number of downward steps, implying that \( l - k \) is even. Thus, \( k \) and \( l \) must share the same parity. Furthermore, the local dynamics ensure that when a minimum site is destroyed, another minimum site with the same parity is created.

As a consequence of property (b), configurations with different parities cannot belong to the same connected component. Starting from any configuration, if we iteratively apply the backward update \( 001 \to 010 \) whenever possible, we eventually reach one of the two \textit{anti-ferromagnetic} states: \( 010101\ldots01 \) for even configurations or \( 101010\ldots10 \) for odd configurations. This implies that configurations with the same parity are connected, which completes the proof of (c). The proof of property (d) is postponed to the next paragraph, in which we will introduce the area of a configuration and show that the integrated current can be identified with this area. The property (d) will be an immediate consequence.
\end{proof}

\noindent
\textbf{The Steady state:} Within each ergodic component, the invariant stationary distribution has a Gibbs form:
\begin{equation}
\label{invmeasure}
\pi(\bfeta ) \;=\; \frac{1}{Z_{N} (q)}\,q^{\,E(\boldsymbol{\eta})}, \quad \text{with } Z_N = \sum_{\bfeta \in \mathcal{A}_N^{(e)}} q^{\,E(\boldsymbol{\eta})}
\end{equation}
And the energy of a configuration $\boldsymbol{\eta}$ is given by:
\begin{equation}
\label{E}
E(\boldsymbol{\eta}) \;=\; \frac{1}{2}\,
\sum_{k=1}^{L}\!\Bigl(h_{k} (\bfeta )\;-\;h_{\text{min}}(\boldsymbol{\eta})
\;-\;\tfrac{1}{2}
\Bigr).
\end{equation}

The half-filling condition ensures that the energy \(E(\boldsymbol{\eta})\) is invariant under cyclic translation.
This energy admits a simple geometric interpretation: if we cyclically rotate the configuration so that it starts at a site of minimal height, the resulting height function \(h_k\) forms a Dyck path \cite{deutsch1999dyck}—a lattice path that remains non-negative and returns to zero at the endpoint. In this representation, \(E(\boldsymbol{\eta})\) corresponds to the area under the Dyck path, measured in units of full lattice squares, as illustrated in Fig.~\ref{fig:Energy}.

From the height representation (Fig.~\ref{fig:height}), it is evident that a forward jump of a particle increases the area by one, while a backward jump decreases it by one. Hence, this area is precisely the total integrated current \(J_t\) defined in Lemma~\ref{lemma1}(d). The maximum area is achieved when the Dyck path forms a single triangular shape, yielding \(E_{\rm max} = \frac{N(N - 1)}{2}\).

Denote by $\bfeta^{k,k+1}$ the configuration after swapping the values of sites $k$ and $k+1$. If this belongs to the same ergodic component as $\bfeta$ it is straightforward to show that:
\begin{equation}
   E(\bfeta^{k,k+1} )=E(\bfeta )+\eta_k -\eta_{k+1}, 
\end{equation}
which implies detailed balance for the distribution $\pi$ Eq. \eqref{invmeasure}.
\begin{equation}
    \pi(\bfeta) w(\bfeta \rightarrow \bfeta') = \pi(\bfeta') w(\bfeta' \rightarrow \bfeta).
\end{equation}

To fully characterize the invariant measure, we must compute the partition function, defined in \ref{invmeasure}.
Interpreting \( E(w) \) as the area under the height profile, this partition function becomes closely related to the area \( q \) - Catalan numbers, defined as \cite{CarlitzRiordan1964}:
\[
C_N(q) = \sum_{w \in \mathcal{D}_N} q^{a(w)},
\]
where \( \mathcal{D}_N \) is the set of Dyck paths of length \( 2N \), and \( a(w) \) denotes the area under the path \( w \). The number of Dyck paths is given by the \( n \)-th Catalan number \( C_N := C_N(1) = \frac{1}{N+1} \binom{2N}{N} \), whereas the size of a single parity configuration space $\mathcal{A}_N^{(e)}$ is  \(\frac{1}{2} \binom{2N}{N} \). Hence, there is no one-to-one correspondence between \( \mathcal{D}_N \) and $\mathcal{A}_N^{(e)}$.
One more reasonably would expect a mapping from $\mathcal{D}_N$ to the set of equivalence classes by cyclic translation of $A_N^{e}$. However, 
if a height profile has several global minima, it corresponds to several distinct
Dyck paths obtained by cyclic rotation. Thus, no such mapping exists either. To account for this degeneracy we
must refine the ordinary $q$-Catalan counting.  The required extension will account for the number of minima and is developed in the next section.

\begin{figure}
    \centering

\begin{tikzpicture}[yscale=1,xscale=1,scale = 0.6]
\draw [fill, color=black!0, fill=black!10] (1,1) -- (2,2) -- (3,1) -- (5,3) -- (7,1) -- (6, 0) -- (5,1)--(4,0)--(3,1)--(2,0)--(1,1);
\draw [thick] (0,0) -- (2,2) -- (3,1) -- (5,3) -- (8,0) -- (9,1) -- (10,0);
\draw [dashed] (0,0) -- (5,5) -- (10,0);
\draw [dashed] (2,0) -- (6,4);
\draw [dashed] (4,0) -- (7,3);
\draw [dashed] (6,0) -- (8,2);
\draw [dashed] (1,1) -- (2,0);
\draw [dashed] (2,2) -- (4,0);
\draw [dashed] (3,3) -- (6,0);
\draw [dashed] (4,4) -- (8,0);
\end{tikzpicture}
    \caption{Example of the height of a configuration of $L=10$ sites and energy $E=4$ represented by the number of shaded squares}
    \label{fig:Energy}
\end{figure}
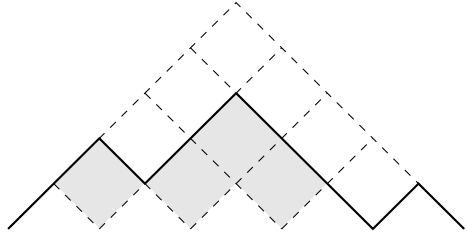

\section[Counting Dyck Paths by Area and Returns]{Counting Dyck Paths by Area and Returns: The $q\!-\!s$ Catalan numbers}
\label{sec:qs-catalan}

The Gibbs weight~\eqref{invmeasure} reduces the statistical mechanics of our
process to a purely combinatorial problem: enumerate height
configurations (equivalently, Dyck paths) with a fugacity \(q\) conjugate to the area \(E\).  A direct mapping to the classical
\(q\)-Catalan numbers \cite{CarlitzRiordan1964} is tempting but incomplete, because
several distinct microscopic configurations correspond to the \emph{same} Dyck
path whenever the height profile contains more than one minimum.  To disentangle
this degeneracy we introduce a second fugacity \(s\) that accounts for the number of  
returns to the axis (i.e.\ the number of global minima), and treat it on an equal
footing with the area fugacity \(q\).  The resulting two–variable polynomials,
which we dub the {\it \(q\)–\(s\) Catalan numbers}, encode the full partition
function of the model and interpolate smoothly between the ordinary Catalan
numbers (\(q=s=1\)) and the Carlitz–Riordan \(q\)-Catalans (\(s=1\)).  We now
formalise this definition and develop the key identities that will underpin the
asymptotic analysis.

Let $w \in \mathcal{D}_N$ be a Dyck path from $(0,0)$ to $(2N,0)$, and denote by $a(w)$ the area under the path. Let $m(w)$ represent the number of returns to the $x$-axis, including the final point $(2N,0)$ but excluding the starting point $(0,0)$.

\noindent
\textbf{Definition:}
We define the joint partition function $G_N(q,s)$, referred to as the $q\!-\!s$ Catalan numbers, by
\begin{equation}
  G_N(q,s) = \sum_{w \in \mathcal{D}_N} q^{a(w)} s^{m(w)}.
\end{equation}

This is a bivariate polynomial of degree $\frac{N(N-1)}{2}$ in $q$ and degree $N$ in $s$. We set $G_0(q,s) = 1$. The first few terms are given by:
$$
1, \; s, \; s (q + s), \; s (q^2 + q^3 + 2 q s + s^2), \; s (q^3 + 
   q^4 (2 + q + q^2) + q^2 (3 + 2 q) s + 3 q s^2 + s^3, \; ...
$$
When $s = 1$, the function reduces to the conventional $q$-Catalan numbers in the sense of Carlitz and Riordan~\cite{CarlitzRiordan1964}:
\begin{equation}
  G_N(q,1) = C_N(q).
\end{equation}
Before diving fully into the properties of the $q\!-\!s$ Catalan numbers, 
Let's establish why they are useful for our study by stating their link to the partition function.

\begin{lemma}
\label{lemma_two}
For every $n\ge 2$, the single parity partition function satisfies
\begin{equation}
Z_N(q)=N\int_{0}^{1}\frac{G_N(q,s)}{s}\,ds .
\label{Z_in_terms_of_G}
\end{equation}
\end{lemma}

\begin{proof}
Assume temporarily that $N$ is a prime number.  
We restrict ourselves to the even ergodic component $\mathcal{A}^{(\mathrm e)}_N$ (the odd component is identically analogous).  Define the equivalence relation $\sim$ on $\mathcal{A}^{(\mathrm e)}_N$ by “cyclic rotation by an even number of sites.”  Each orbit under this relation contains exactly $N$ configurations, except for the anti-ferromagnetic state, which is invariant and forms an orbit of size one.

Similarly, we define an equivalence relation on the set $\mathcal{D}_N$ of Dyck paths: two Dyck paths are equivalent if one can be obtained from the other via cyclic rotation.  
Clearly, the number of minima $m(w)$ of a Dyck path $w$ is constant within each orbit and equals the orbit's size except for the anti-ferromagnetic Dyck path, which forms a singleton orbit. There is an obvious bijection between the set of orbits of $\mathcal{D}_N$ and those of $\mathcal{A}_N^{(e)}$.  The ratio of the size of an orbit of $A_N^{e}$ to its corresponding orbit in $\mathcal{D}_N$ is trivially $N/m(w)$ with $w$ being any represntative path. If $N$ is not prime, certain orbits may exhibit additional symmetry of order $d \mid N$, but these symmetries are matched between corresponding orbits in $\mathcal{D}_N$ and $\mathcal{A}_N^{(e)}$, preserving thus the size ratio.
Hence,
\[
  Z_N(q) =\sum_{\bfeta \in A_N^{e}} \,q^{a(\bfeta )}
  =
  \sum_{w \in \mathcal{D}_N} \frac{N}{m(w)}\,q^{a(w)}.
\]
Finally, using the identity $\int_{0}^{1} s^{m-1} \, ds = \frac{1}{m}$ and the definition
  $G_N(q,s)$
we obtain the desired identity~\eqref{Z_in_terms_of_G}.
\end{proof}

\begin{lemma}
     The $q\!-\!s$  Catalan numbers verify the following recurrence relation for $N \geq 1$

\begin{equation}\label{rec}
    G_{N}(q,s) = \sum_{0 \leq k \leq N-1} s q^{N-k-1} G_{k}(q,s) G_{N-k-1}(q,1) 
\end{equation}
\end{lemma}

\begin{proof}
First notice that both of the properties $a$ and $m$ are additive with respect to concatenation of Dyck paths. Following the method of \cite{deutsch1999dyck}, define $\hat{\mathcal{D}}_N = \{ w \in \mathcal{D}_N; m(w) =1  \}$, the set of elevated Dyck paths that don't return to zero except at the extremities. The corresponding partition function would be:
\begin{equation}
  \hat{G}_N(q,s) := \sum_{w \in \hat{\mathcal{D}}_N} q^{a(w)} s^{m(w)}   = s q^{N-1} G_{N-1}(q,1)   
\end{equation}
A recurrence is established by cutting each path at the first return to zero, which can happen at a site $ 0\leq k\leq N-1$
\begin{equation}
    G_{N}(q,s) = \sum_{0 \leq k \leq N-1} G_{k}(q,s)\hat{G}_{N-k}(q,s) =  \sum_{0 \leq k \leq N-1} s q^{N-k-1} G_{k}(q,s) G_{N-k-1}(q,1)
\end{equation}
\end{proof}
To analyse the large-\(N\) growth of \(G_N(q,s)\) we can reformulate the recurrence
\eqref{rec} in terms of the ordinary generating function defined as:
\begin{equation}
\Omega(q,s,z)\;=\;\sum_{n\ge 0} G_N(q,s)\,z^{N}.
\end{equation}
Multiplying \eqref{rec} by \(z^{N}\) and summing over \(N\ge 0\) yields the
\(q\)-functional (or \(q\)-algebraic) relation.
\begin{equation}\label{q_fun}
  \Omega(q,s,z)-1 \;=\; s\,z\,\Omega(q,s,z)\,\Omega\!\bigl(q,1,qz\bigr),
\end{equation}
whose dominant singularities in the complex \(z\)-plane determine the
asymptotic form of \(G_N(q,s)\) via the transfer theorems of analytic
combinatorics \cite{flajolet2009analytic,barbe2013q,barbe2020q}.

\paragraph*{Special cases.}
For \(q=s=1\) equation~\eqref{q_fun} collapses to a quadradic equation giving the generating function of the (Non deformed) Catalan numbers:
\begin{equation}
    \mathcal{G}_{C}(z) :=
\Omega(1,1,z)\;=\;\frac{1-\sqrt{1-4z}}{2z}.
\label{generating_C}
\end{equation}
Setting \(q=1\) but keeping \(s\) arbitrary produces a closed expression,
\begin{equation}\label{Omega1sx}
\Omega(1,s,z)\;=\;\frac{2}{\,2+s\!\left(\sqrt{1-4z}-1\right)},
\end{equation}
With the first few terms of its expansion:
\[
  \Omega(1,s,z)=1+s z+(s+s^{2})z^{2}+(2s+2s^{2}+s^{3})z^{3}+O(z^4).
\]
The coefficient of \(s^{m}z^{N}\) counts Dyck paths of semi-length \(N\) that
return to the axis exactly \(m\) times, hence, these coefficients are simply the classical Ballot numbers, and $G_N(q,s)$ reduces to Ballot polynomials:
\[
  G_{N}(1,s)=\sum_{m=0}^{N}
             \binom{2N-m-1}{\,N-m\,}\frac{m}{N}\,s^{m}\qquad(N>0).
\]
\paragraph*{Hypergeometric representation:}
Introduce the \(q\)-hypergeometric series
\begin{equation}
      F(q,z)\;=\;\sum_{k\ge 0}\frac{(-1)^{k}q^{k^{2}-k}z^{k}}{(q)_{k}},
  \quad\text{with }(q)_{k}=\prod_{j=1}^{k}(1-q^{j}),\;(q)_{0}=1 .
  \label{F}
\end{equation}
\begin{lemma}\label{Ram}
The generating function \(\Omega(q,s,z)\) can be written as
\begin{equation}\label{hyper}
  \Omega(q,s,z)\;=\;
  \frac{F\!\bigl(q,q z\bigr)}
       {s F\!\bigl(q,z\bigr)+(1-s)\,F\!\bigl(q,q z\bigr)}.
\end{equation}
\end{lemma}
\begin{proof}
The proof is based on the q-functional equation \ref{q_fun} in addition to using a standard property of the q-hypergeometric function, namely: $z F(q, q^2 z) = - F(q,z) + F(q,q z) $, for which we recall the proof:
\begin{align}
        z F(q, q^2 z) &= \sum_{k \geq 0} (-1)^k q^{k^2 - k -2} (q^2 z)^{k+1} /(q)_{k} = \sum_{k \geq 0} (-1)^k q^{k^2 - k -2}
    \frac{(q^2 z)^{k+1} (1-q^{k+1})}{(q)_{k+1}} \\
   &= \sum_{k \geq 0} (-1)^k q^{k^2 - k -2}
    \frac{(q^2 z)^{k+1} }{(q)_{k+1}}
    -
 \sum_{k \geq 0} (-1)^k q^{k^2 - k -2}
    \frac{(q^2 z)^{k+1} q^{k+1}} {(q)_{k+1}}
\end{align}

Since the term associated with $k=0$ is equal for both sums:
\begin{align}
        z F(q, q^2 z) &= -\sum_{k \geq 0} (-1)^k q^{k^2 - 3 k}
    \frac{(q^2 z)^{k} }{(q)_{k}}
    +
 \sum_{k \geq 0} (-1)^k q^{k^2 - 3 k}
    \frac{(q^2 z)^{k} q^{k}} {(q)_{k}} \\
 &= -\sum_{k \geq 0} (-1)^k q^{k^2 -  k}
    \frac{z^{k} }{(q)_{k}}
    +
 \sum_{k \geq 0} (-1)^k q^{k^2 -  k}
    \frac{(q z)^{k}} {(q)_{k}} =
     -
    F(q,z)
    +
 F(q,q z) 
\end{align}
It's straightforward to show that \ref{q_fun}.
can be satisfied by 

\begin{equation}
  \Omega(q,s,z)\;=\;
  \frac{F\!\bigl(q,qz\bigr)}
       {F\!\bigl(q,qz\bigr)-sz\,F\!\bigl(q,q^2 z\bigr)}.
\end{equation}
Applying the previous property once more leads to  \ref{hyper} 
\end{proof}
Note that in the limit $s=1$ Eq. \ref{hyper} reduces the ratio $\frac{F(q,qz)}{F(q,z)}$, which is the generating  function of the q-Catalan number (See chapter 7 of \cite{andrews1998theory})

\subsection{Asymptotic behavior of $q\!-\!s$ Catalan numbers:}
the asymptotic behaviour of $q\!-\!s$ Catalan numbers depends critically  on the value of $q$ and eventually on $s$. It is convenient to distinguish three regimes. $0<q<1$, $q=1$ and $q>1$.

\noindent
$\boxed{q < 1}$
For each fixed \(0<q<1\), the series \ref{F}
is entire in \(z\), and the meromorphic generating function \ref{hyper}
can only have poles at the real zeros of the denominator. Denote $\Phi(z)\;=\;s\,F(q,z)-(1-s)\,F(q,qz)$. Restricted on the real axis, $\Phi$ is a real function and satisfies:
$$\Phi(0)=s>0,
\quad
\lim_{z\to\infty}\Phi(z)=-\infty,
\quad
\Phi'(z)
=s\,\partial_zF(q,z)+(1-s)\,q\,\partial_zF(q,qz)<0\ \text{ for all } z>0,$$
so \(\Phi\) is strictly decreasing and crosses zero exactly once.
Denote its unique simple root by \(\rho(q,s)>0\). 
Residue calculus allows to retrieve the asymptotic behavior:
\[
G_N(q,s)\sim\;
\frac{-F(q,\rho)}{\rho\,\Phi'(\rho)}\;\rho^{-N},
\]
Note that there is a critical behavior for a set of values $(q,s)$ defined by the curve
\(\displaystyle s_c(q)=\frac{F(q,q)}{F(q,q^2)}\)
that
separates two regimes:
\[
\begin{cases}
s>s_c(q): & \rho<1,\ G_N\text{ grows exponentially},\\
s<s_c(q): & \rho>1,\ G_N\text{ decays exponentially},\\
s=s_c(q): & \rho=1,\ G_N\sim\bigl(F(q,1)/\Phi'(1)\bigr)
\end{cases}
\]
This critical curve is defined for a domain $0 < q <  q_0$ for such that $s_c(q) > 0$ and $q_0$ is the solution of the equation $F(q,q) = 0$, numerically $q_0 \approx 0.58$. For $q>q_0$, $\rho<1$ and $G_N$ grow exponentially.

\noindent
$\boxed{q = 1}$ The asymptotic behavior can be easily obtained from the generating function Eq. \ref{Omega1sx}. The dominant singularity of \(\Omega\) depends on \(s\), producing three
distinct asymptotic regimes:
\[
  G_N \;\sim\;
  \begin{cases}
    \displaystyle
    \dfrac{s}{(2-s)^{2}}\,
    \dfrac{4^{N}}{\sqrt{\pi}\,N^{3/2}},
      & 0<s<2,\\[1.1em]
    \displaystyle
    \dfrac{4^{N}}{\sqrt{\pi N}},
      & s=2,\\[1.1em]
    \displaystyle
    \dfrac{s-2}{s-1}\,
    \bigl(\tfrac{s^{2}}{s-1}\bigr)^{N},
      & s>2.
  \end{cases}
  \qquad(N\to\infty)
\]

\begin{itemize}[leftmargin=1.8em]
\item \textbf{Sub-critical} \(\mathbf{(0<s<2)}\):  
  The branch point of \(\sqrt{1-4z}\) at \(z=\tfrac14\) is closest to the
  origin.  A standard transfer theorem then yields the “Catalan’’
  amplitude \(4^{N}N^{-3/2}\).

\item \textbf{Critical point} \(\mathbf{(s=2)}\):  
  The simple pole in the denominator collides with the branch point,
  increasing the algebraic order to \(N^{-1/2}\).

\item \textbf{Super-critical} \(\mathbf{(s>2)}\):  
  The pole moves inward to  
  \(
    z_\star=(s-1)/s^{2}<\tfrac14
  \),
  outstripping the branch point; its residue
  \(
    (s-2)/(s-1)
  \)
  sets the prefactor of the pure exponential
  \(
    (s^{2}/(s-1))^{N}.
  \)
\end{itemize}

\noindent
$\boxed{q > 1}$ The  asymptotic behavior of $G_N(q,s)$ grows superexponentially as $q^{\binom{N}{2}}$
as it becomes identical to the q-Catalan numbers, up to a constant factor $s$, as stated in the following proposition:
\begin{proposition}
\label{prop1}
For every fixed $q>1$ and every fixed $s>0$,
\begin{equation}
\label{eq:Gs-over-Cs}
\lim_{N \to \inf}
\frac{G_{N}(q,s)}{C_{N}(q)} = s.
\end{equation}
\end{proposition}

\begin{proof}
Let $w\in D_{N}$ be a Dyck path of semi-length $N$ and with number of returns to zero $m(w)$.  The maximal area attainable at semi-length $N$ is
$a_{\max}(N)=\tfrac{N(N-1)}{2}$, achieved uniquely by a full triangle path for which $m=1$.
Write the returns of a general path as $0=r_{0}<r_{1}<\dots<r_{m}=2N$ and set $k_{i}=\tfrac12(r_{i}-r_{i-1})$; then $\sum_{i=1}^{m}k_{i}=N$ and
\begin{equation}
a_{\max}(N)-a(w) \ge \frac12\Bigl[N(N-1)-\sum_{i=1}^{m}k_{i}(k_{i}-1)\Bigr]
\ge 
(m-1) (N - m/2)
\label{eq:area-loss}
\end{equation}
where the last inequality uses $k_{1}=N-m+1$ and $k_{2}=\dots=k_{m}=1$, the configuration that maximaizes the area of a Dyck path with $m$ returns to zero.
To prove the proposition, we proceed by decomposing $G_N(q,s)$ into the path with $m=1$, which can be mapped to elevated Dyck path of semi-length $N-1$,
and those for $m \geq 2$
\begin{equation}
     G_{N}(q,s)=
  s q^{N-1}
  C_{N-1}(q)+ \sum_{m=2}^{N}s^{m} C_{N,m}(q)
  \quad
\qquad
 C_{N,m}(q) := \sum_{w:\,m(w)=m}q^{a(w)}.
 \label{eq:decompose}
\end{equation}

Remember that the the asymptotics for $C_N(q)$ is well known for $q>1$ \cite{furlinger1985q, flajolet2009analytic}, and given by $C_N(q) \sim
\frac{s}{\phi(1/q)} q^{\binom{N}{2}} 
$, with $\phi$ is the q-pochammer Euler function. This is enought to show that the first term of Eq. behaves as $s C_N(q)$:
\begin{equation}
 \lim_{N \rightarrow \infty}   \frac{q^{N-1} C_{N-1}(q)}{C_{N}(q)} = 1
\end{equation}
Now what we need to show is that the second term in Eq. \ref{eq:decompose} is negligable comparing to $C_N(q)$.
For that purpose, we need to estimate a bound for the partition function restriced on Dyck paths with exactly $m$ returns to zero $C_{N,m}$. For $m=1$ this simply corresponds to the elevated Dyck paths: $\hat{C}_N(q) := C_{N,1}(q)$.
Again $\hat{C}_N(q)$ behaves as $q^{\binom{N}{2}}$ up to a constant that depends only on $q$
More conveniently, for $N$ large enough, there exists a function $K(q)$ such that:
$$\hat{C}_N(q) \leq K(q) q^{\binom{N}{2}}.$$
We can see a Dyck path that returns to zero at points  $0=r_{0}<r_{1}<\dots<r_{m}=2N$ as a concatination of elevated Dyck paths of semi lengths $k_1, k_2,...,k_m$, which allows to write $C_{N,m}$ as:
\begin{equation}
    C_{N,m}(q) = \sum_{k_1+..+k_m = N}  \hat{C}_{k_1}(q)... \hat{C}_{k_m}(q)
    \leq
     \sum_{k_1+..+k_m = N}  K(q)^m
     q^{\binom{k_1}{2}}...q^{\binom{k_m}{2}}
\end{equation}
To adresss the behaviour the second term of Eq. \ref{eq:decompose}, we need to provide an estimation for the ratio:
\begin{equation}
\frac{C_{N,m}(q) }{C_{N}(q) }
\leq
K(q)^m
\binom{N-1}{m-1}
q^{-(m-1)(N-m/2)}
\lesssim N^m q^{-(m-1)(N-m/2)}
\end{equation}
Now we can bound  the second term in Eq. \ref{eq:decompose}:
\begin{align}
   \sum_{m=2}^{N}s^{m}
    \frac{C_{N,m}(q) }{C_{N}(q)}
    \lesssim &
    \sum_{m=2}^{N} s^{m} N^m q^{-(m-1)N + m^2/2}
        =
    q^{N}\sum_{m=2}^{N} (s N q^{-N})^m q^{m^2/2} \\
&=  q^{N-N^2/2} 
    \sum_{m=2}^{N} (s N)^m q^{(N-m)^2/2}  =
  q^{N-N^2/2}  (s N)^N
    \sum_{k=0}^{N-2} (s N)^{-k} q^{k^2/2}
\end{align}

Checking the behaviour of the ratio of two consequtive terms of sum in the previous expression, one can verifies that this series is dominated by the last terms
\begin{equation}
   \sum_{k=0}^{N-2} (s N)^{-k} q^{k^2/2}
=
(sN)^{-(N-2)}\,q^{\frac{(N-2)^2}{2}}\,(1+o(1)),
\qquad (N\to\infty). 
\end{equation}
Which finally leads to the desired estimation:
\begin{equation}
   \sum_{m=2}^{N}s^{m}
    \frac{C_{N,m}(q) }{C_{N}(q)}
    \lesssim
(sN)^{2}\,q^{ 2 - N}\,(1+o(1))
\end{equation}
This shows that this terms converge to zero as $N$ goes to infinity, which proves the proposition.
\end{proof}

\begin{corollary}
Using the Asymptotics of $C_N(q)$ for $q>1$ \cite{furlinger1985q, flajolet2009analytic},
The behaviour of  $q\!-\!s$  Catalan numbers  is given by:
    \begin{equation}
G_{N}(q,s)=\frac{s}{\phi(1/q)} q^{\binom{N}{2}} \bigl(1+o(1)\bigr),\quad q > 1, \quad N\to\infty
\end{equation}
with $\phi$ is the 
q-Pochhammer Euler function
$\phi(u) = \prod_{i=1}^{\infty} (1 - u^i)$
\label{cor1}
\end{corollary}

\subsection{Asymptotic behaviour of the partition function}
The various thermodynamic quantities rely on estimating the asymptotic behaviour of the partition function for large $N$. Combining the recurrence \ref{rec} and lemma \ref{Z_in_terms_of_G}, one can have simple computational access to $Z_N(q)$, allowing for numerical checks. Analytically, the behaviour of $Z_N(q)$ can be computed using results of the previous section and will be critically dependent on $q$:

\noindent
$\boxed{q > 1}$ Using proposition \ref{prop1} and lemma \ref{Z_in_terms_of_G}, we obtain for $q>1$:
\begin{equation}
    \lim_{N \rightarrow \infty} \frac{Z_N(q)}{n C_N(q)} = 1
\end{equation}
Which leads to 
\begin{equation}
    Z_N(q) \sim \frac{N}{\phi(1/q)}
    q^{\binom{N}{2}}, 
    \quad
    \phi(u)= \prod_{i=1}^{\infty} (1-u^{i})
    \label{asympZN}
\end{equation}

It's useful to note that as $q$ approaches $1$, the coefficient diverges, and it has an asymptotic behavior:

\begin{equation}
    1/\phi(1/q) = \sqrt{\frac{2\pi}{t}} \exp\left(-\frac{\pi^2}{6 t} + \frac{t}{24}\right) + o(1), \quad t = \ln(q), \quad q \rightarrow 1^+
     \label{asympEuler}
\end{equation}

This is based on a known asymptotic of the q-pochhammer function $\phi$ by Watson \cite{watson1936final}. With the additional second term inside the exponential, the formula is an excellent numerical approximation for finite values of $q$, that are not necessarily close to unity.

\noindent
$\boxed{q = 1}$ We have: $Z_N(1) = \frac{1}{2} \binom{2N}{N}$, which yields:
\begin{equation}
    Z_N(1) =  \frac{n+1}{2}C_N(1) \sim {\frac {4^{N}}{{2\sqrt {\pi N }}}}
\end{equation}

\noindent
$\boxed{q < 1}$ We show that in this case, the partition function has the same growth rate as $C_N(q)$ for $q<1$ but with a distinct prefactor.

\begin{proposition}\label{prop:Zn_q<1}
Let $q \in (0,1)$ and denote by $\xi(q)$ the largest real positive zero of
$F(q,1/z)$, where $F$ is defined in~\eqref{F}. The partition function for a single–parity sector is given by:
\begin{equation}
   Z_N(q)=\xi(q)^N \,\bigl[1+o(1)\bigr],
   \qquad N\to\infty,
   \label{eq:Zn_asymp}
\end{equation}
with an {\it exact} prefactor~$1$.
\end{proposition}

\begin{proof}
Set $\rho_0 = 1/\xi$ the smallest positive real zero of $F(q,z)$.
For fixed $s \in (0,1]$ the denominator of
$\Omega(q,s,z)$ has a simple pole at $z=\rho(q,s)$, and we previously showed that for $q<1$:
\[
   G_N(q,s)=B(q,s)\,\rho(q,s)^{-N}\,[1+o(1)],
\]
where $B(q,s)=F(q,q\rho)/[-\rho\,\partial_z \Phi |_{z=\rho}]$ and
$\Phi(z,s)=sF(q,z)+(1-s)F(q,qz)$.  Since $\rho(q,s)$ decreases strictly in
$s$ and $\rho(q,s)\searrow \rho_0$ as $s\nearrow1$, the integrand
$N\,G_N(q,s)/s$ in~\eqref{Z_in_terms_of_G} is maximised near $s=1$.  
Write $s=1-\zeta/N$; then
\[
   Z_N(q)=
   \int_{0}^{N}
     B\!\Bigl(1-\tfrac{\zeta}{N}\Bigr)\,
     \rho\!\Bigl(1-\tfrac{\zeta}{N}\Bigr)^{-N}\,
     \Bigl(1-\tfrac{\zeta}{N}\Bigr)^{-1} d\zeta .
\]
Expanding $\rho$ at $s=1$ gives
$\rho(1-\zeta/N)^{-N}=\rho_0^{-N}\exp\{- \zeta\,\lambda\}[1+O(N^{-1})]$ with
$\lambda=-\rho'(1)/\rho_0>0$.  Keeping the leading terms and applying
Watson’s one-sided Laplace lemma \cite{de2014asymptotic,wong2001asymptotic} yields:
\begin{equation}
    Z_N(q)=\frac{B_0}{\lambda}\rho_0^{-N}[1+o(1)], \qquad B_0 := \frac{-F(q,q\rho_0)}{\rho_0F'(q,\rho_0)}
\end{equation}

Finally, differentiating $\Phi(\rho(q,s),s)=0$ at $s=1$ shows
$\rho'(1)=F(q,q\rho_0)/F'(q,\rho_0)$, so $B_0/\lambda=
\rho_0\,F(q,q^{2}\rho_0)/F(q,q\rho_0)$.  
The $q$-difference identity
$zF(q,q^{2}z)=F(q,qz)-F(q,z)$, evaluated at $z=\rho_0$, gives
$\rho_0F(q,q^{2}\rho_0)=F(q,q\rho_0)$; hence the prefactor equals~$1$.
\end{proof}

\begin{remark}
As $q$ increases from $0$ to $1$ the root $\rho_0(q)$ decreases from $1$
to $1/4$, interpolating smoothly between $Z_N(0)=O(1)$ for strong backward bias
and $Z_N(1)=\frac12\binom{2N}{N}\sim4^{N}/(2\sqrt{\pi N})$ for the no bias symmetric regime.
\end{remark}

\noindent
\textbf{Numerical Checks:}
Since the space of configurations grows exponentially with the system size, a brute force enumeration of the partition function is accessible only for very small systems. However, the recurrence relation \ref{rec} allows for easy access to the $q\!-\!s$ Catalan numbers at a cost that is of the order of $N^2$ (which can be refined to $N \ln (N)$ using fast Fourier transform).
Numerical integration of \ref{Z_in_terms_of_G} allows then to have access to $Z_N$. In Fig. \ref{fig:beta}. $Z_N$ is numerically compared with the asymptotic predictions previously derived. Note that asymptotics rely on extracting numerically the growth rate $\xi(q)$ that is the largest real root of the equation:
\begin{equation}\label{betaq}
    \sum_{k=0}^{\infty} \frac{(-1)^k q^{2 k^2-k}}{(q;q)_k (\xi(q))^k} = 0 
\end{equation}
Since the terms of this series decay super-exponentially, the first few terms yield very good numerical approximations. 
\begin{figure}
  \centering
\includegraphics[width=0.48\linewidth]{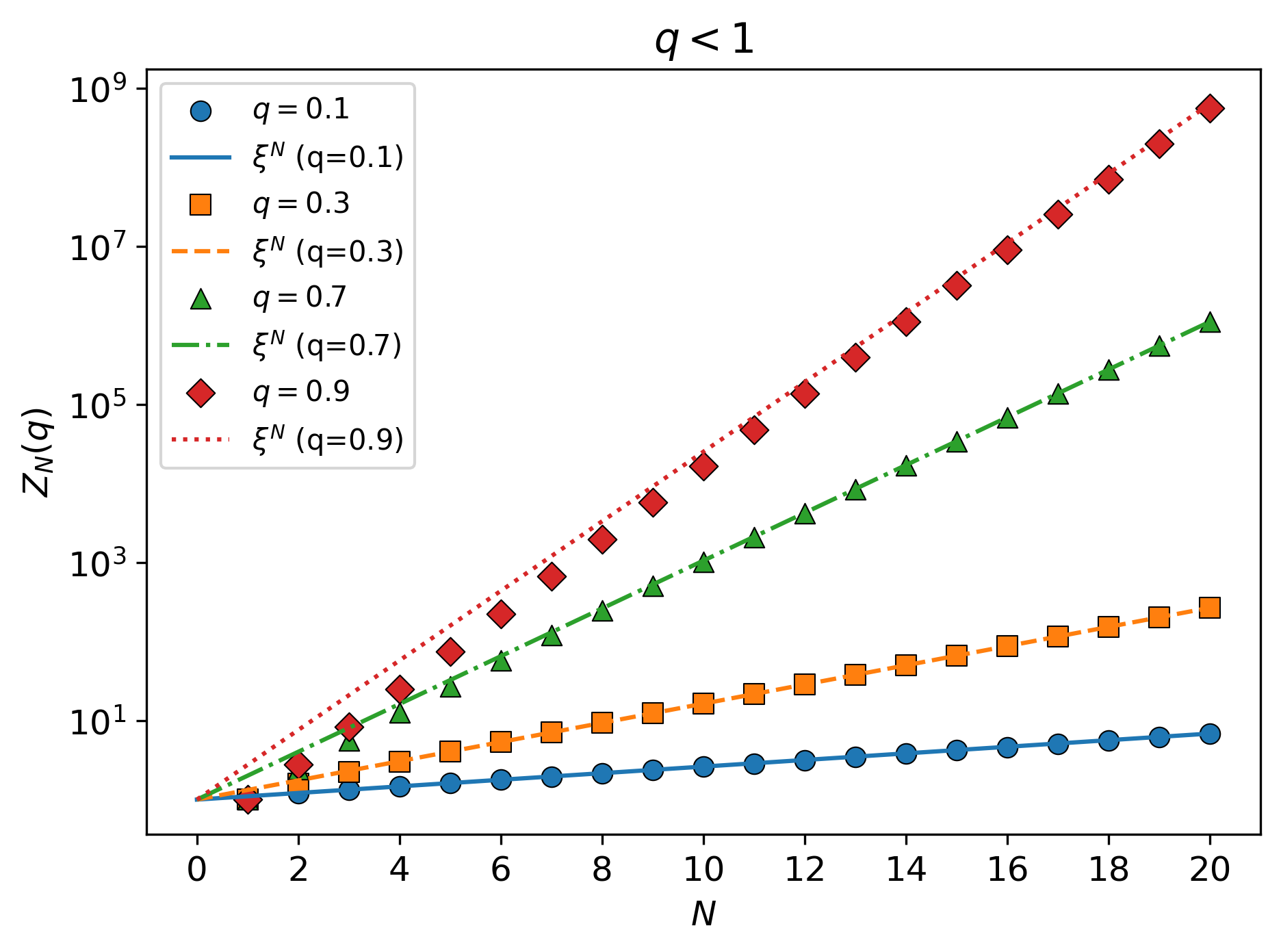}\hfill
\includegraphics[width=0.48\linewidth]{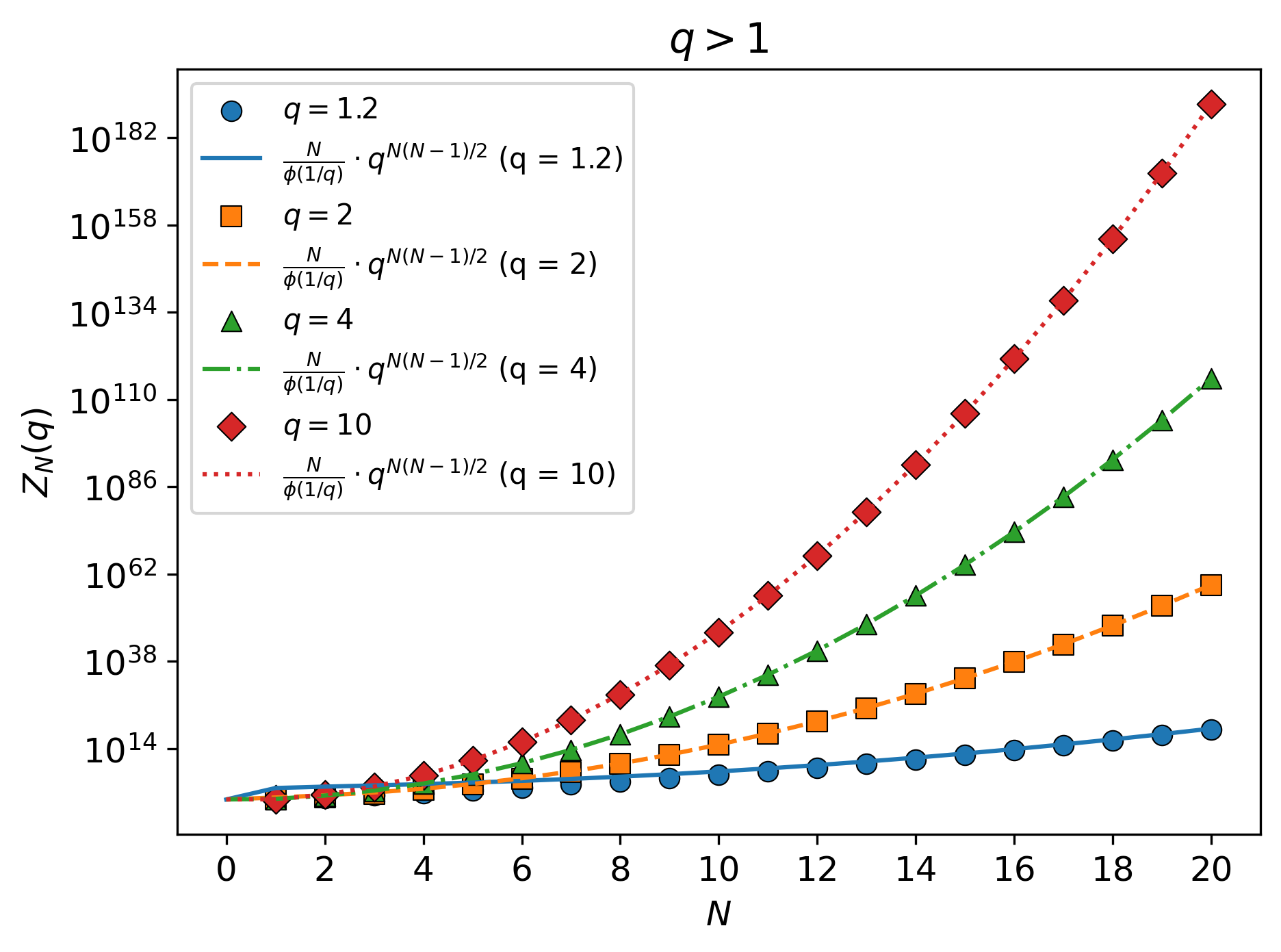}
  \caption{Asymptotic behaviour of $Z_N(q)$ illustrating that the free energy is extensive for $q<1$ (left) ans super extensive for $q>1$ (right). The dashed lines represent the theoretical asymptotic curves while the shapes represent the exact values computed numerically}
  \label{fig:beta}
\end{figure}

\section{Mixed-Order Phase Transition and related thermodynamic Observables}
\label{sec:mixed}
\subsubsection{Free energy and order parameter}
The free energy is defined as $F_N = \frac{1}{\ln(q)} \ln(Z_N)$ (the inverse temperature here would be $-\ln(q)$)
From the previous section, we have:
\begin{equation}
    F_N \sim 
        \begin{cases}
      \frac{1}{2} N^2 & \text{if } q > 1\\
       \frac{ \ln(\xi(q))}{\ln(q)} N & \text{if } q < 1
    \end{cases}
\end{equation}
The first surprising observation is that the free energy is not extensive for $q>1$. This is a breakdown of additive thermodynamics where the energy of two subsystems is supposed to be the sum of the energy of the parts. This is a signature of long-range correlations: In a fully connected graph (the simplest example is the fully connected Ising model) the number of interacting pairs would be of the order of $N^2$.

The previous scaling suggests defining
an order parameter $ m(q) = \lim_{N \rightarrow \infty}
    \frac{F_N}{N^2}$

\begin{equation}
m(q) =
        \begin{cases}
       \frac{1}{2}  & \text{if } q > 1\\
        0 & \text{if } q < 1
    \end{cases}
\end{equation}

The jump in the order parameter is characteristic of a first-order phase transition. However, to dig into the actual nature of the transition, we need to look at the behavior of the derivatives of the free energy.
\subsubsection{Expectation and Variance of the area}
The expectation and variance of the area with respect to the invariant measure \ref{invmeasure} are respectively given by:
\begin{equation}
\langle A_N(q) \rangle  = q \frac{d}{dq} \ln(Z_N(q)), \quad
\operatorname{Var}[A_N(q)]
   \;=\;
   q\,\frac{d}{dq}\!\Bigl(q\,\frac{d}{dq}\ln Z_{N}(q)\Bigr).
\end{equation}

In the thermodynamic language, these quantities are the internal energy and susceptibility respectively.
for $q=1$, these quantities can be computed exactly. for that purpose one needs to compute the first and second derivatives of the partition function at $q=1$ $\dot{Z}_N(1)$, $\ddot{Z}_N(1)$, which is done in the Appendices A and D. These computations give:

\begin{equation}
\langle A_N(1) \rangle =
    \frac{N}{2} \left(\frac{\sqrt{\pi} N! }{\Gamma\left(\frac{1}{2} + n\right)} - 2\right), 
    \quad
    \operatorname{Var}[A_N(1)] = 
    \frac{1}{12} N^2 \left(2 + 10n - \frac{3 \pi \Gamma(1 + N)^2}{\Gamma\left(\frac{1}{2} + n\right)^2}\right)
    \label{exp_var_q_one}
\end{equation}
One can easily extract the asymptotic behavior of these formulas. Additionally, one can derive the asymptotics of $A_N$ and $\operatorname{Var}[A_N]$ in case $q\neq 1$ from those of $Z_N$ derived in the previous section. These derivations give:
\begin{equation}
\langle A_N(q) \rangle
    \sim 
        \begin{cases}
      \frac{1}{2} N^2 & \text{if } q > 1\\
       \frac{\sqrt{\pi}}{2}N^{\frac{3}{2}}  & \text{if } q = 1\\
        \frac{q\xi'(q)}{\xi(q)} N & \text{if } q < 1
    \end{cases},
    \quad
    \operatorname{Var}[A_N(q)]
    \sim 
        \begin{cases}
      \frac{2 \pi^2 - 3\ln(q)}{6(\ln(q))^3} & \text{if } q > 1\\
       (\frac{5}{6} - \frac{\pi}{4})N^3   & \text{if } q = 1 \\
       \kappa(q) N & \text{if } q < 1
    \end{cases}
\end{equation}

With
\begin{equation}
  \kappa(q) = 
\frac{q \left( -q \, \xi'(q)^2 + \xi(q) \left( \xi'(q) + q \, \xi''(q) \right) \right)}{\xi(q)^2} 
\end{equation}

\paragraph*{Physical interpretation.}
The asymptotic behaviour of the three observables
($F_N$, $\langle A_N\rangle$, $\operatorname{Var}[A_N]$)
identifies two distinct equilibrium phases separated by a mixed–order
transition at the critical asymmetry \(q_c=1\).
\begin{itemize}
\item \textbf{Condensed (jammed) phase, \(\boldsymbol{q>1}\).}\\[2pt]
  The mean area reaches the maximal value
  \(\langle A_N\rangle\sim \tfrac12 N^{2}\),
  while the variance stays bounded,
  \(\operatorname{Var}[A_N]\sim
    \tfrac{2\pi^{2}-3\ln q}{6(\ln q)^{3}}\).
  Hence almost all \(N\) particles merge into a single compact cluster
  of length \(\Theta(N)\); only its rear edge fluctuates.
  The free energy grows super-extensively,
  \(F_N\sim\tfrac12 N^{2}\), signalling the breakdown of additive
  thermodynamics typical of long-range interacting (“fully-connected’’)
  systems and a spontaneous breaking of translational invariance accompanied by weak ergodicity breaking: although all fully clustered positions are equiprobable in the steady state, the time needed to dismantle a macroscopic cluster and re-form it elsewhere scales as \(\sim q^{N^{2}}\) (i.e., super-exponential in \(N\)). Hence, on any practical time scale the system remains trapped near one cluster location; see also \cite{grosskinsky2025long} for a discussion of this metastability.
\item \textbf{Fluid phase, \(\boldsymbol{q<1}\).}\\[2pt]
The additivity of thermodynamics is restored. The mean area is linear in the system size, indicating that the typical configuration is antiferromagnetic. Thermal fluctuations around this state are short-ranged, as expected for equilibrium systems with finite-range interactions away from criticality. Clusters remain microscopic, of order \(O(1)\), in direct analogy with the facilitated-exclusion limit \(q=0\). The free-energy density \(F_N/N\) is analytic, as in an ordinary short-range lattice gas.

\item \textbf{Critical point, \(\boldsymbol{q=1}\).}\\[2pt]
  The order parameter
  \(m(q)=\displaystyle\lim_{N\to\infty}F_N/N^{2}\)
  jumps from \(0\) to \(\tfrac12\) (first-order signature),
  yet the susceptibility diverges:
  \(\operatorname{Var}[A_N(1)]\sim
    (\tfrac56-\tfrac\pi4)\,N^{3}\).
  The mean area scales as
  \(\langle A_N(1)\rangle\sim
    \tfrac{\sqrt\pi}{2}\,N^{3/2}\),
  reproducing the Brownian-excursion law for Dyck paths, with an average hight of order $\sqrt{N}$.
  Together with the non-extensive correction
  \(F_N(1)\propto N^{3/2}\),
  these facts establish a mixed-order transition.
\end{itemize}
Mixed-order phase transitions are well documented \cite{bar2014mixed,korbel2025microscopic} and have also appeared in kinetically constrained lattice gases and jamming settings \cite{toninelli2006jamming,toninelli2007jamming,gross1985mean,schwarz2006onset}. The present transition, however, is unusual: it occurs in one dimension, at equilibrium, is tuned by a hopping-rate asymmetry rather than by density, and the “frozen’’ configurations are only {\it metastably} jammed (not strictly frozen in finite time). The transition likewise shares qualitative features with a glass transition \cite{ritort2003glassy}: a quench from \(q<1\) to \(q>1\) produces a long-lived metastable clustered state whose relaxation to the true stationary state is exceedingly slow and displays aging. In this analogy, \(q=1\) plays the role of a melting point separating fluid-like and cluster-dominated regimes. Finally, to the best of our knowledge, our model exhibits a trajectory-space active–inactive phenomenology closely related to that reported for stochastic Fredkin chains \cite{causer2022slow}.
\section{Weakly asymmetric constrained exclusion process}
\label{sec:weakly}
To probe the behavior near the transition point $q=1$, we let the asymmetry depend on the system size,
\begin{equation}
    q_N \;=\; \exp\!\bigl(b/N^{\alpha}\bigr),
\end{equation}
with real parameters $b$ and $\alpha$. When $\alpha=0$ this reduces to a strong, $N$-independent asymmetry. In the regime $\alpha>0$ we do not have full analytic control of thermodynamic quantities and will rely on heuristics supported by numerics. The asymptotics derived in the previous section apply for fixed $q$; a natural question is for which $(b,\alpha)$ those formulas remain valid under the scaling above.

We restrict to $q>1$ ($b>0$), where our approach seems to provide reliable results. Inserting $q_N$ into the fixed-$q$ asymptotics of $Z_N$ from Eq.~\eqref{asympZN} yields
\begin{equation}
    Z_{q_N}\;\sim\; \frac{N}{\phi(1/q_N)}\, q_N^{\binom{N}{2}}, 
    \qquad b>0.
    \label{asym_Weak}
\end{equation}
We expect \eqref{asym_Weak} to remain valid for small $\alpha$, up to some threshold $\alpha<\alpha^\star$. Conversely, for sufficiently large $\alpha$ the system should behave as at $b=0$, starting beyond a second threshold $\alpha>\tilde\alpha$. Our numerics indicate a single crossover at
\[
    \alpha^\star=\tilde\alpha=1.
\]
In particular, \eqref{asym_Weak} appears accurate up to $\alpha\lesssim 1$, with a distinct behavior for $\alpha>1$. We conjecture closed-form expressions including prefactors.

\paragraph*{Thermodynamic quantities in the weakly asymmetric regime.}
Write $Z_N(b,\alpha):=Z_N(q_N)$. Since $\ln q_N=b/N^{\alpha}$, we have
\[
    N^{\alpha}\,\frac{\mathrm d}{\mathrm db}
    \;=\;
    \frac{\mathrm d}{\mathrm d(\ln q)} 
    \;=\;
    q\,\frac{\mathrm d}{\mathrm dq}.
\]
With this convention, we define
\begin{equation}
    F_N(b,\alpha)=\frac{N^{\alpha}}{b}\ln Z_N(b,\alpha), 
    \quad 
A_N(b,\alpha)=N^{\alpha}\,\frac{\mathrm d}{\mathrm db}\ln Z_N(b,\alpha),
    \quad
V_N(b,\alpha)=N^{2\alpha}\,\frac{\mathrm d^2}{\mathrm db^2}\ln Z_N(b,\alpha).
\end{equation}
Using Eq.~\eqref{asympEuler} to evaluate \eqref{asym_Weak}, we are led to the following asymptotic behavior:
\begin{equation}
\label{eq:Weakly}
\begin{aligned}
F_N(b,\alpha) &\sim 
\begin{cases}
\frac{1}{2}N^2, & \text{if }\alpha<1,\\
\frac{\ln 4}{b}\,N^{1+\alpha}, & \text{if }\alpha\ge 1,
\end{cases}\\[0.6ex]
A_N(b,\alpha) &\sim 
\begin{cases}
\frac{1}{2}N^2, & \text{if }\alpha<1,\\
g(b)\,N^2, & \text{if }\alpha=1,\\
\frac{\sqrt{\pi}}{2}\,N^{3/2}, & \text{if }\alpha>1,
\end{cases}
\qquad
V_N(b,\alpha) &\sim 
\begin{cases}
\frac{2\pi^2-3b}{6b^3}, & \text{if }\alpha=0,\\
\frac{\pi^2}{3b^3}\,N^{3\alpha}, & \text{if }0<\alpha<1,\\
\left(\frac{5}{6}-\frac{\pi}{4}\right)N^3, & \text{if }\alpha\ge 1.
\end{cases}
\end{aligned}
\end{equation}

Here $g(b)$ is an $O(1)$ prefactor at the marginal scaling $\alpha=1$; for large $b$ we obtain
\[
    g(b)\sim \tfrac{1}{2}-\frac{\pi^2}{6\,b^2}.
\]
These predictions have been checked against numerical simulations, testing both scaling exponents and prefactors. Fig. \ref{weaklyFree},\ref{weaklyArea} and \ref{weaklyVar}. The agreement is generally good, with visible deviations at small $b$—most notably for $\alpha>1$ in $A_N$ and $V_N$. We attribute these discrepancies to finite-size effects and subleading terms from the $\phi(\cdot)$ expansion; a finer analysis is needed to resolve them.

\begin{figure}[h]
    \centering
  \begin{minipage}{0.48\textwidth}
    \centering
    \makebox[\linewidth]{\small $\alpha<1$}
    \vspace{0.4em}
    \includegraphics[width=0.8\linewidth]{"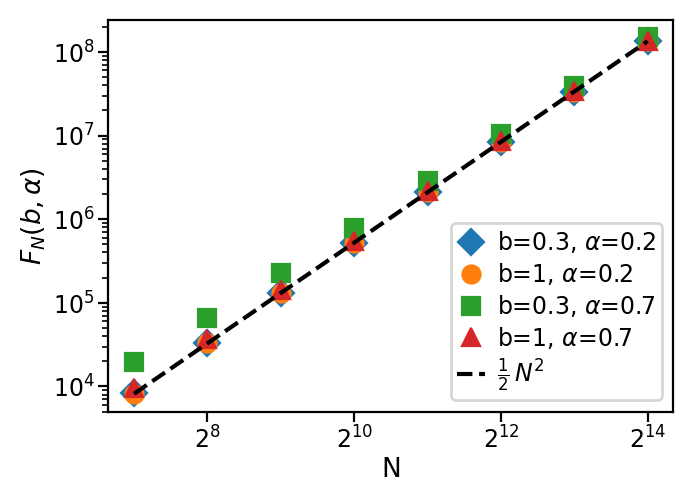"}
  \end{minipage}\hfill
  \begin{minipage}{0.48\textwidth}
    \centering
    \makebox[\linewidth]{\small $\alpha\geq 1$}
    \vspace{0.4em}
    \includegraphics[width=0.8\linewidth]{"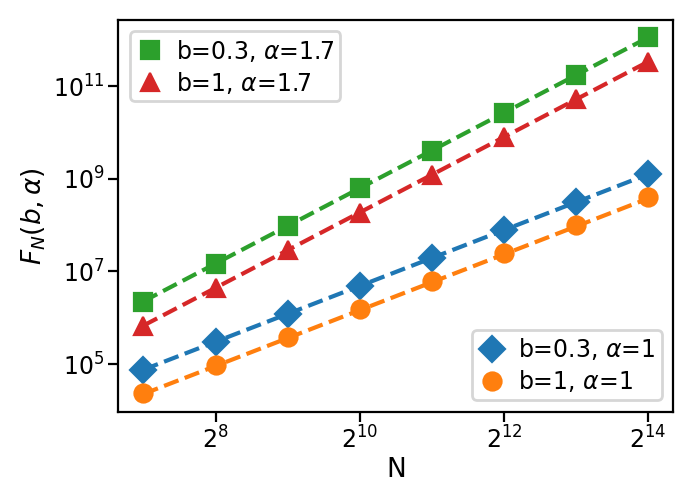"}
  \end{minipage}
    \caption{Weakly asymmetric regime: free-energy scaling on log–log axes. Dashed lines show the conjectured asymptotics from Eq.~\eqref{eq:Weakly}, while markers denote numerical data. For $\alpha<1$ , the growth expoend and prefector are constants for $\alpha\ge 1$, the growth exponent depends on $\alpha$ and the prefactor is a function of $b$.}
    \label{weaklyFree}
\end{figure}

\begin{figure}[h]
    \centering
  \centering
  \begin{minipage}{0.32\textwidth}
    \centering
    \makebox[\linewidth]{\small$\alpha<1$}
    \vspace{0.4em}
    \includegraphics[width=\linewidth]{"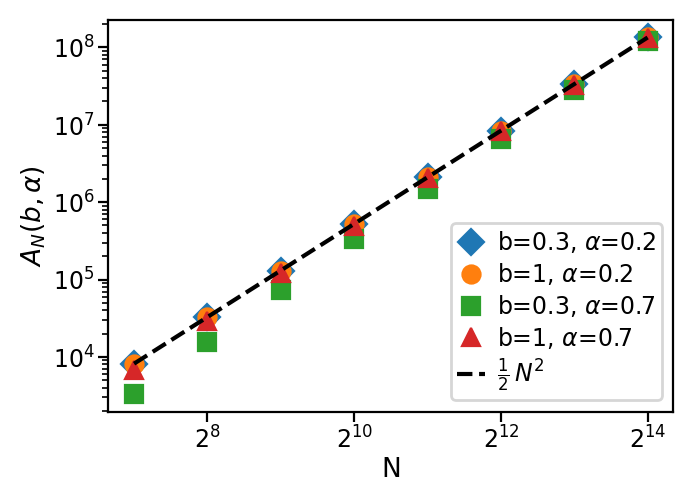"}
  \end{minipage}\hfill
  \begin{minipage}{0.32\textwidth}
    \centering
    \makebox[\linewidth]{\small $\alpha=1$}
    \vspace{0.4em}
    \includegraphics[width=\linewidth]{"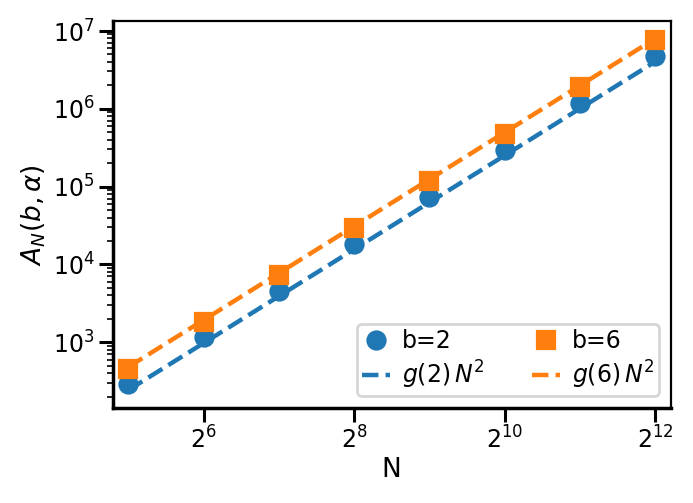"}
  \end{minipage}\hfill
  \begin{minipage}{0.32\textwidth}
    \centering
    \makebox[\linewidth]{\small $\alpha>1$}
    \vspace{0.4em}
    \includegraphics[width=\linewidth]{"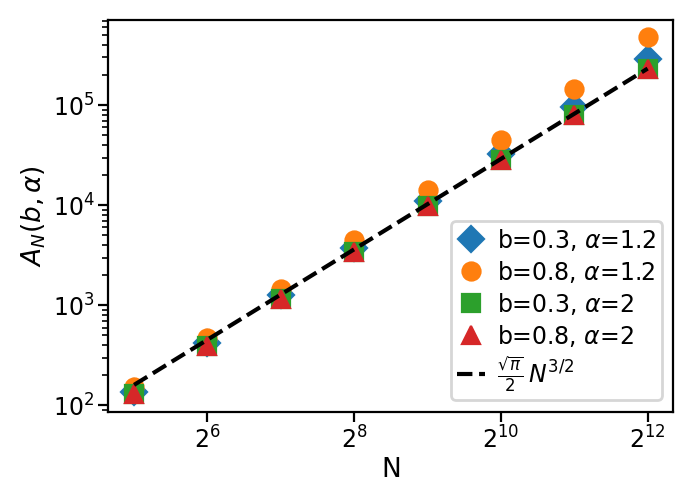"}
  \end{minipage}
    \caption{Weakly asymmetric regime: average-area scaling on log–log axes. Dashed lines show the conjectured asymptotics from Eq.~\eqref{eq:Weakly}; markers denote numerical data. For $\alpha<1$, $A_N\sim \tfrac{1}{2}N^2$. For $\alpha=1$, $A_N\sim g(b)\,N^2$ with a $b$-dependent prefactor known asymptotically for large $b$ ($g(b)\sim \tfrac{1}{2}-\pi^2/6b^2$). For $\alpha>1$, $A_N\sim (\sqrt{\pi}/2)\,N^{3/2}$.}
    \label{weaklyArea}
\end{figure}

\begin{figure}[h]
    \centering
  \begin{minipage}{0.32\textwidth}
    \centering
    \makebox[\linewidth]{\small  $\alpha=0$}
    \vspace{0.4em}
    \includegraphics[width=\linewidth]{"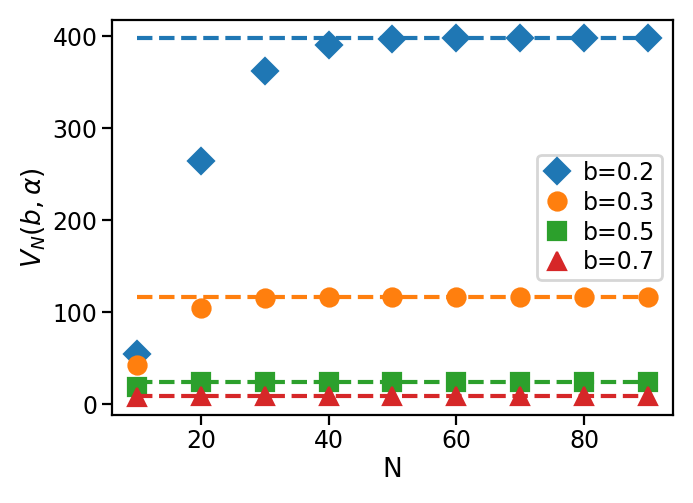"}
  \end{minipage}\hfill
  \begin{minipage}{0.32\textwidth}
    \centering
    \makebox[\linewidth]{\small  $\alpha<1$}
    \vspace{0.4em}
    \includegraphics[width=\linewidth]{"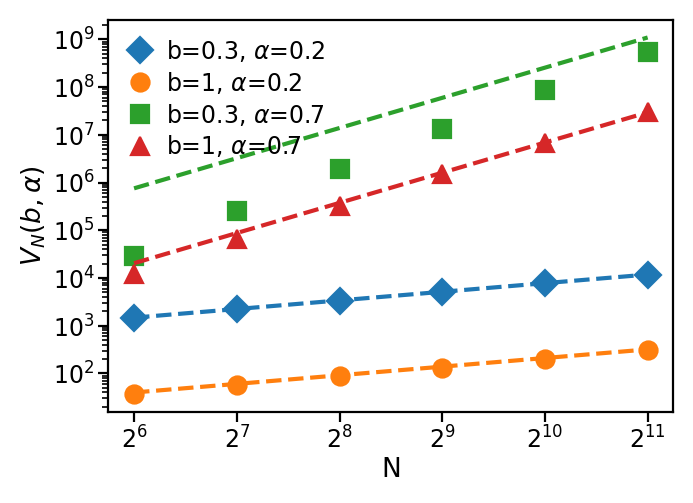"}
  \end{minipage}\hfill
  \begin{minipage}{0.32\textwidth}
    \centering
    \makebox[\linewidth]{\small $\alpha\geq 1$}
    \vspace{0.4em}
    \includegraphics[width=\linewidth]{"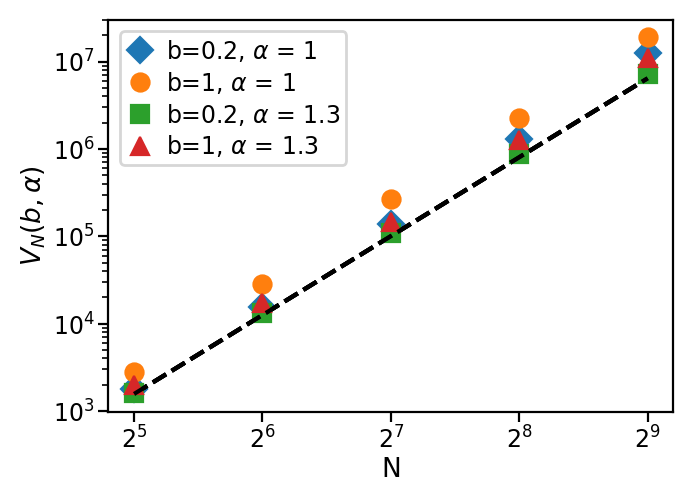"}
  \end{minipage}
    \caption{Weakly asymmetric regime: variance of the area. Dashed lines show the conjectured asymptotics from Eq.~\eqref{eq:Weakly}; markers denote numerical data. The middle and right panels use log–log axes (the left panel, $\alpha=0$, uses linear axes). For $0<\alpha<1$, $V_N\sim (\pi^2/3b^3)\,N^{3\alpha}$: the growth exponent depends on $\alpha$ and the prefactor on $b$. For $\alpha\ge 1$, $V_N\sim \bigl(5/6-\pi/4\bigr)N^3$, independent of both $\alpha$ and $b$; for $\alpha=0$, $V_N\sim (2\pi^2-3b)/(6b^3)$.}
    \label{weaklyVar}
\end{figure}

\section{More than half filling} \label{sec:more}

\begin{figure}[h]
    \centering
\begin{tikzpicture}[yscale=1,xscale=1,scale = 0.6]
\begin{scope}[yshift = -40]
    \draw [thick]
    (0,0) -- (14,0);
	
\foreach \i in {0,...,14}
	{
		\draw (\i,0) -- (\i,0.1) ;
	}
	
\foreach \i in {0,1,3,5,7,9,11,13}
	{
		\node at (\i+0.5,0.5){};
        \draw [very thick,fill] (\i+0.5,0.3) circle (5pt);
        }
\foreach \i in {4,6,8,10,12}
	{
        \draw [thick] (\i+0.5,0.3) circle (5pt);
        }
\draw [thick, red] (2+0.5,0.3) circle (5pt);
\foreach \i in {0,1,13}
	{
		\node at (\i+0.5,0.5){};
        \draw [very thick,red,fill] (\i+0.5,0.3) circle (5pt);
        }
        
\end{scope}
\end{tikzpicture}
    \caption{Introducing an additional particle into a half-filled system creates a frozen segment consisting of three particles and one hole (red). The remaining part of the system forms an active segment (black), corresponding effectively to a half-filled system with closed boundary conditions.}
    \label{fig:more}
\end{figure}
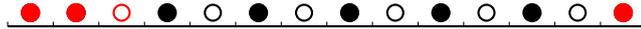

We conclude this work by giving some insights into the dynamics for systems with particle density greater than one-half, consider the following simplified scenario. Starting from a half-filled system prepared in an antiferromagnetic configuration, we introduce an additional particle into one of the empty sites. As illustrated in Fig.~\ref{fig:more}, this insertion creates a frozen segment consisting of three particles and one hole. The remainder of the system, termed the \textit{active segment}, comprises $M = N - 3$ particles occupying a region of $2M$ sites. Notably, the frozen segment effectively imposes closed boundary conditions on the active segment, simplifying its invariant measure compared to the periodic case. Indeed, since the height function for the active segment never becomes negative, the set of allowed configurations maps precisely onto the set of Dyck paths of semi-length $M$.
If \( w \) denotes a configuration within the active segment, its stationary probability is given by
\begin{equation}
\pi(w) = \frac{q^{E(w)}}{C_{M}(q)},
\end{equation}
where the energy \( E(w) \) associated with configuration \( w \) is defined as
\begin{equation} 
E(w) = \frac{1}{2}\sum_{k=1}^{L}\left(h_k(w)-\frac{1}{2}\right),
\end{equation}
and the normalization factor \( C_{M}(q) \) corresponds to the \( q \)-Catalan number:
\begin{equation}
C_{M}(q) = \sum_{w \in D_{M}} q^{E(w)},
\end{equation}
where \( D_{M} \) denotes the set of Dyck paths of semi-length \( M \). 

More generally, a system with particle number \( N > \frac{L}{2} \) consists of multiple alternating frozen and active segments. The maximum number of active segments is determined precisely by the number of extra particles beyond half-filling, given by \( 2N - L - 1 \).

\section{Conclusion}
We studied an exclusion process with a minimal kinetic constraint that profoundly alters the steady-state properties of the system and gives rise to an equilibrium-like phase transition at half filling. 
Building on \cite{grosskinsky2025long}, we focused on this transition and uncovered the combinatorial underpinning of the invariant measure: a natural generalization of $q$-Catalan numbers with a second deformation that counts Dyck-path returns to zero in addition to the area. This structure yields closed expressions for thermodynamic observables (including finite-size corrections) and reveals a mixed-order transition—discontinuous in suitable order parameters yet accompanied by diverging response—akin to phenomena previously reported only in more elaborate settings \cite{bar2014mixed}. Significantly, this transition arises in one dimension
for {\it short-range} interactions, unlike other transitions between clustered
and disordered or hyperuniform states that require long-range interactions
in one dimension \cite{Kare17,Kare22} or more than one conserved type of particles \cite{Evan98b,Lahi00,Chak16}.

Despite this progress, several questions remain open. The most immediate concerns the invariant measure at less than filling, which appears highly non-trivial. In particular, for densities below one half and $q>1$ the stationary current exhibits negative differential mobility (a non-monotonic dependence on $q$), reminiscent of higher-dimensional kinetically constrained lattice gases \cite{benichou2014microscopic,sellitto2008asymmetric}; its appearance in one dimension seems peculiar and, to the best of our knowledge, lacks a general framework.

A second direction is the hydrodynamic limit. A weakly asymmetric scaling $q=\exp(-b/N)$ is natural for deriving a macroscopic equation via non-gradient methods à la Varadhan–Quastel \cite{kipnis2013scaling}.
However, so far, only the limit $q\to 0$ where the process reduces
to the totally asymmetric facilitated exclusion process has been studied \cite{Scho04,Blon20,Blon21}

Third, our model at half filling shares striking similarities with the three-species ABC model at equal densities and cyclic hopping \cite{clincy2003phase}, which also displays a transition between homogeneous and clustered phases. Heuristically, the kinetic constraint here may play a role analogous to the additional species in the ABC dynamics; clarifying this correspondence would be valuable.

Finally, it would be important to delineate the relation to stochastic Fredkin spin chains \cite{causer2022slow}, which involve four-site interactions yet display related phenomenology. At half filling, $q$-Catalan combinatorics again emerge. Establishing precise mappings or shared universality classes would substantially broaden the scope of the present results.

We expect the exact combinatorial framework developed here to serve as both a benchmark for simulations and a springboard for addressing these questions.

\backmatter
\bmhead{Acknowledgements}

We would like to thank Stefan Grosskinsky for useful discussions.
This work is financially supported by the grants
FCT/Portugal through project UIDB/04459/2020 (DOI: 10.54499/UIDP/04459/2020), as well as through grants 2020.03953.CEECIND (DOI: 10.54499/2020.03953.CEECIND/CP1587/CT0013) and 2022.09232.PTDC (DOI: 10.54499/2022.09232.PTDC).
In addition, A.Z. gratefully acknowledges the financial support of the ANR-PRME Uniopen, project (ANR-22-CE30-0004-01).

\begin{appendices}

\section{Computing the  derivatives of the partition function at $q=1$}
\label{derivatives_Z}
In order to compute the expectation and the variance of the area at $q=1$, Eq. \ref{exp_var_q_one}, one needs to extract the first and second derivatives of the partition function at $q=1$. For notation simplicity, denotes theses derivatives $\dot{Z}$ and $\ddot{Z}$ respectively. These are integer sequences that we will extract from their generating functions, whose 
computations will require adapting Lemma \ref{lemma_two} and the q-alebraic relation Eq. \ref{q_fun}, and needs additionally the first two derivatives of $q-$ Catalan numbers at $q=1$, which are computed in Appendix \ref{Derivatives_C}. Derive the q-algebraic relation \ref{q_fun} with respect to $q$:
\begin{equation}
    \partial_q \Omega(q,s,x) =  s x(
    \partial_q \Omega(q,s,x) \Omega(q,1,q x)
    +
    \Omega(q,s,x) \partial_q \Omega(q,1,q x)
    +
    x \Omega(q,s,x) \partial_x  \Omega(q,1,q x)
    )
\end{equation}
Evaluating for $q=1$, we have:
\begin{equation}
\partial_q \Omega(1,s,x) =  s x(
    \partial_q \Omega(1,s,x) \mathcal{G}_{C}(x)
    +
    \Omega(1,s,x) (\mathcal{G}_{\dot{C}}(x)
    +
    x \dot{C}(x) )
    ) 
\end{equation}
Where $\mathcal{G}_{\dot{C}}$ is the generating function of the the derivative of q-Catalan numbers at $q=1$ which is computed in Appendix \ref{Derivatives_C} and given by Eq. \ref{generating_C_dot} and $\mathcal{G}_{C}(x)$ is the generating 
function of the conventional Catalan number defined in Eq. \ref{generating_C}, $\dot{C}(x)$ is its derivative and
 $\Omega(1,s,x)$ is given by Eq. \ref{Omega1sx}. This allows us to express $\partial_q \Omega(1,s,x)$ in terms of known functions:
\begin{equation}
    \partial_q \Omega(1,s,x) = \frac{s x \Omega(1,s,x) \left(\mathcal{G}_{\dot{C}}(x) + x \dot{C}(x) \right)}{1 - s x \mathcal{G}_{C}(x)},
\end{equation}
and can be expressed explicitly as:
\begin{equation}
    \partial_q \Omega(1,s,x) = \frac{2 s x (-1 + \sqrt{1 - 4 x} + 4 x)}{(2 + s (-1 + \sqrt{1 - 4 x}))^2 (1 - 4 x)^{\frac{3}{2}}}.
    \label{qOmega(1,s,x)}
\end{equation}
Denote the generating function of $\dot{Z}_N$ as $\mathcal{G}_{\dot{Z}}$, one can use Eq. \ref{Z_in_terms_of_G} to express it as:
\begin{equation}
    \mathcal{G}_{\dot{Z}}(x) = \int_{0}^{1} \frac{x}{s} \partial_x \partial_q \Omega(1,s,x)
    ds
\end{equation}
Evaluating using Eq. \ref{qOmega(1,s,x)}, we can find an explicit expression
\begin{equation}
     \mathcal{G}_{\dot{Z}}(x) = \frac{(1 - \sqrt{1 - 4x})x}{(1 - 4x)^2}
\end{equation}
An infinite series expansion allows us to retrieve $\dot{Z}_N$
\begin{equation}
    \dot{Z}_N(1) = n \cdot 2^{2N-2} - (2N-1) \cdot \binom{2N-2}{N-1}
    \label{dot{Z}1}
\end{equation}
The first few terms of this sequence are given by:
$$0, 2, 18, 116, 650, 3372, ....$$
This sequence is identifiable by  OEIS number A153338 \cite{OEIS-A015333}.
The first term in Eq. \ref{dot{Z}1} dominates the second, meaning that $\dot{Z}_N(1)$ behaves asymptotically for large $N$ as
\begin{equation}
    \dot{Z}_N(1) \sim n \cdot 4^{N-1}
\end{equation}

\noindent
\paragraph*{Second order derivative:}
Denote $\mathcal{G}_{\ddot{Z}}$ the generating function of the second order derivative at $q=1$.. This generating function can be expressed using \ref{Z_in_terms_of_G} as:
\begin{equation}\label{gen_ddot{Z}}
G_{\ddot{Z}}(x) =
  \int_{0}^{1} \frac{x}{s} \partial_{x}\partial_{qq} \Omega(1,s,x)
    ds
\end{equation}
To compute $\partial_{x}\partial_{qq} \Omega(1,s,x)$ we derive Eq. \ref{q_fun} twice with respect to $q$ and evaluate for $q=1$
\begin{align*}
&2 s x \left( x \partial_{x} \Omega(1, 1, x)+\partial_{q} \Omega(1, 1, x) \right) \partial_{q} \Omega(1, s, x)
+
\left(-1 + s x \Omega(1, 1, x)\right) \partial_{qq} \Omega(1, s, x) \\
&+
s x \Omega(1, s, x) \left( x^2 \partial_{xx}\Omega(1, 1, x) + 2 x \partial_{q x} \Omega(1, 1, x) + \partial_{qq} \Omega(1, 1, x) \right)
= 0
\end{align*}
Which leads to:
\begin{equation}
\partial_{qq} \Omega(1,s,x) = 
    \frac{8 s x \left(\left(1 - 4 x  + 10 x^2\right) + \sqrt{1 - 4 x} 
    \left(-1 + 2 x + s (-1 + -2x + 3 x^2)\right)
    \right)}{\left(2 + s \left(-1 + \sqrt{1 - 4 x}\right)\right)^3 (1 - 4 x)^{\frac{5}{2}}}
\end{equation}
The generated function for $\ddot{Z}_N(1)$ can be computed using Eq. \ref{gen_ddot{Z}}
\begin{equation}
\mathcal{G}_{\ddot{Z}}(x) =
    \frac{x \left( 10 x^2 - 
    2x+ 3 - (3+4x)\sqrt{1 - 4 x})
    \right)}{(1 - 4 x)^{\frac{7}{2}}}
\end{equation}
Which allows though full series expansion to find $\ddot{Z}_N$ explicitly:
\begin{equation}
\ddot{Z}_N :=
    \frac{1}{3} 4^{(-1 + N)} \left(-3 n (1 + 2 N) + \frac{(6 + n (7 + 5 N)) \Gamma\left(\frac{1}{2} + n\right)}{\sqrt{\pi} \Gamma(N)}\right)
\end{equation}
For large $N$, this sequence has the asymptotic:
\begin{equation}
    \ddot{Z}_N \sim \frac{5 N^{\frac{5}{2}}}{3 \sqrt{\pi}} \cdot 4^{N-1}
\end{equation}

\section{Computing the derivatives of q-Catalan numbers at $q=1$ }
\label{Derivatives_C}
The first and second order derivatives at $q=1$ of q-Catalan are integer sequences that we denote $\dot{C}_N$ and $\ddot{C}_N$ respectively. We will compute their generating functions from which we will etract their exact expressions.
Denote the generating functions by $\mathcal{G}_{\dot{C}}$ and $\mathcal{G}_{\ddot{C}}$ respectively.
We have $\mathcal{G}_{\dot{C}}(x) := \partial_q f(1,x)$, with $f(q,x) := \Omega(q,1,x)$ the generating function of q-Catalan numbers, which satisfies:
\begin{equation}\label{q_alg}
    f(q,x) - 1 = x f(q,x)f(q,qx)
\end{equation}
Deriving with respect to q:
\begin{equation}
    \partial_q f(q,x) = x f(q,qx) \partial_q f(q,x)
    +
    x f(q,x)(\partial_q f(q,qx)
    +x \partial_x f(q,qx)  )
\end{equation}
At $q=1$, we know $f(1,x)$ which is simply the generating function of the conventional Catalan numbers, denoted $\mathcal{G}_{C}(x)$, and we can easily compute $\partial_x f(1,x)$
\begin{equation}
    \mathcal{G}_{C}(x) := f(1,x) ²=\frac {1-{\sqrt {1-4x}}}{2x}, \qquad
   \partial_x f(1,x) =
   \frac{-2x - \sqrt{1 - 4x} + 1}{2x^2\sqrt{1 - 4x}}
\end{equation}
Now the generating function $\mathcal{G}_{\dot{C}}(x)$ satisfies
\begin{equation}
    \mathcal{G}_{\dot{C}}(x) = 2 x \mathcal{G}_{C}(x) \mathcal{G}_{\dot{C}}(x)
    +
    x^2 \mathcal{G}_{C}(x) \partial_x f(1,x)
\end{equation}
Which yields
\begin{equation}
 \mathcal{G}_{\dot{C}}(x) = \frac{1 - \sqrt{1 - 4x} - 3x + x\sqrt{1 - 4x}}{2x - 8x^2}
\end{equation}
The first few terms of this generating function are:
\begin{equation}
    \mathcal{G}_{\dot{C}}(x) = x^2 + 7 x^3 + 37 x^4 + 176 x^5 + 794 x^6 + O(x^7)
    \label{generating_C_dot}
\end{equation}
This is the generating function of the sequence:
\begin{equation}
    \dot{C}_N = 2^{2N+1} - \binom{2N+3}{N+1} + \binom{2N+1}{N}
\end{equation}

\noindent
\paragraph*{Second order derivative:}
Derive the q-algebraic equation \ref{q_alg} to the second order with respect to $q$ and evaluate at $q=1$
\begin{equation}
\begin{split}
\mathcal{G}_{\ddot{C}}(x) = \partial_{qq} f(1,x) = x \Big( & 2 \partial_q f(1,x) \left(x \partial_x f(1,x) + \partial_q f(1,x)\right) \\
& + f(1,x) \left(x^2 \partial_{xx} f(1,x) + 2 \left(x \partial_{qx} f(1,x) + \partial_{qq} f(1,x)\right)\right) \Big)
\end{split}
\end{equation}
To find $\partial_{qq}f(1,x)$, we need, $\partial_{xx} f(1,x)$ and $\partial_{qx} f(1,x)$:
\begin{equation}
    \partial_{xx} f(1,x) = \frac{d^2 }{dx^2} \mathcal{G}_{C}(x) = -\frac{{1 - \sqrt{1 - 4x} + 2x \left(-3 + 2\sqrt{1 - 4x} + 3x\right)}}{{\left(1 - 4x\right)^{\frac{3}{2}}x^3}}
\end{equation}
\begin{equation}
\partial_{qx} f(1,x) = g'(x) = \frac{-1 + \sqrt{1 - 4x} + 2x\left(4 - 3\sqrt{1 - 4x} + (-6 + \sqrt{1 - 4x})x\right)}{2(1 - 4x)^2 x^2}
\end{equation}
Now we have all the ingredients to compute $\mathcal{G}_{\ddot{C}}(x)$:
\begin{equation}
   \mathcal{G}_{\ddot{C}}(x) = \frac{1 - \sqrt{1 - 4x} (1+2x) - 6 x^2}{(1 - 4x)^{\frac{5}{2}}} 
\end{equation}
It can be shown by computing an arbitrary order derivative that this is the generating function of the sequence:
\begin{equation}
    \ddot{C}_N = \frac{-1}{6} \cdot 2^{2N} \left( 6 + 9n + \frac{\left(6 + N(19 + 5N)\right) \Gamma\left(N + \frac{1}{2}\right)}{\sqrt{\pi} \cdot N!} \right)
\end{equation}
The asymptotic of this sequence for large $N$ is given by:
\begin{equation}
    \ddot{C}_N \sim \frac{5 \cdot 2^{2N} \cdot N^{\frac{3}{2}}}{6 \cdot \sqrt{\pi}}
\end{equation}

\end{appendices}

\bibliography{f}

\end{document}